\documentclass[11pt]{article}
\usepackage[utf8]{inputenc}

\usepackage{graphicx}
\usepackage{xcolor}
\usepackage{amsmath,amsthm,amssymb}
\usepackage[margin=1in]{geometry}
\usepackage[pdftex,colorlinks=true,linkcolor=blue,citecolor=blue,urlcolor=black]{hyperref}
\usepackage{comment}
\usepackage[normalem]{ulem}
\usepackage{cite}
\usepackage{bm}
\usepackage{tcolorbox}
\usepackage{algorithm}
\usepackage{algorithmic}

\usepackage{multirow}

\def\ba{\begin{array}}
\def\ea{\end{array}}

\DeclareMathOperator*{\argmin}{arg\,min}

\def\0{{\bf 0}}
\def\a{{\bf a}}
\def\b{{\bf b}}

\def\L{{\mathcal L}}
\def\K{{\mathcal K}}
\def\e{{\bf e}}

\def\x{{\bf x}}

\newcommand{\ket}[1]{| #1 \rangle}
\newcommand{\bra}[1]{\langle #1|}

\newcommand{\be}{\begin{equation}}
\newcommand{\ee}{\end{equation}}
\newcommand{\bea}{\begin{eqnarray}}
\newcommand{\eea}{\end{eqnarray}}
\newcommand{\bes}{\begin{equation*}}
\newcommand{\ees}{\end{equation*}}
\newcommand{\beas}{\begin{eqnarray*}}
\newcommand{\eeas}{\end{eqnarray*}}

\makeatletter
\newtheorem*{rep@theorem}{\rep@title}
\newcommand{\newreptheorem}[2]{%
\newenvironment{rep#1}[1]{%
 \def\rep@title{#2 \ref{##1} (restated)}%
 \begin{rep@theorem}}%
 {\end{rep@theorem}}}
\makeatother

\allowdisplaybreaks

\newtheorem{thm}{Theorem}
\newtheorem*{thm*}{Theorem}

\newtheorem{lem}[thm]{Lemma}
\newtheorem*{lem*}{Lemma}

\newtheorem{prop}[thm]{Proposition}
\newtheorem{defn}[thm]{Definition}

\newreptheorem{thm}{Theorem}
\newreptheorem{lem}{Lemma}
\newreptheorem{cor}{Corollary}
\newtheorem{prob}{Problem}

\makeatletter
\newenvironment{breakablealgorithm}
  {
   \begin{center}
     \refstepcounter{algorithm}
     \hrule height1pt depth0pt \kern3pt
     \renewcommand{\caption}[2][\relax]{
       {\raggedright\textbf{\ALG@name~\thealgorithm} ##2\par}%
       \ifx\relax##1\relax 
         \addcontentsline{loa}{algorithm}{\protect\numberline{\thealgorithm}##2}%
       \else 
         \addcontentsline{loa}{algorithm}{\protect\numberline{\thealgorithm}##1}%
       \fi
       \kern3pt\hrule\kern3pt
     }
  }{
     \kern3pt\hrule\relax%
   \end{center}
  }
\makeatother

\usepackage{pgfplots}

\usepackage{authblk}
\usepackage[title]{appendix}

\title{ Quantum speedup of leverage score sampling and its application }

\author{Changpeng Shao\thanks{changpeng.shao@bristol.ac.uk}}
\affil{School of Mathematics, University of Bristol, UK}

\date{\today}

\begin{document}

\maketitle

\begin{abstract}

Leverage score sampling is crucial to the design of randomized algorithms for large-scale matrix problems, while the computation of leverage scores is a bottleneck of many applications. In this paper, we propose a quantum algorithm to accelerate this useful method. The speedup is at least quadratic and could be exponential for well-conditioned matrices. We also prove some quantum lower bounds, which suggest that our quantum algorithm is close to optimal.
As an application, we propose a new quantum algorithm for rigid regression problems with vector solution outputs. It achieves polynomial speedups over the best classical algorithm known. In this process, we give an improved randomized algorithm for rigid regression.

\vspace{.2cm}

\end{abstract}

\section{Introduction}

Randomized algorithms for large-scale matrix problems (e.g., linear regressions and low-rank approximations) have received great attention in recent years \cite{mahoney2011randomized,woodruff2014sketching}. Sketching and sampling are basic tools. Among all the sketching and sampling techniques, (statistical) leverage score sampling plays a crucial role in many applications, e.g., see \cite{drineas2006sampling,chepurko2022near,kacham2022sketching,clarkson2017low,chowdhury2018iterative,mahoney2011randomized,sarlos2006improved,mahoney2009cur,papailiopoulos2014provable,drineas2008relative,drineas2012fast,drineas2010effective,drineas2011faster}. The best random sampling algorithms use these scores to construct an importance sampling distribution to sample with respect to \cite{mahoney2011randomized}.

Leverage scores measure the extent of the correlation between singular vectors of a matrix and the standard basis. They can be computed as the squared norm of the rows of the matrix containing the top left (or right) singular vectors. 
The leverage score sampling technique aims to reduce a large-scale matrix problem to a small-scale one by sampling certain rows (or columns) from the input matrix according to the distribution defined by the leverage scores. Moreover, from the solution of the small-scale problem, we can recover a high accurate solution of the original problem.
Classically, to perform leverage score sampling we need to compute all the leverage scores first. This turns out to be a bottleneck in many applications.
The best classical algorithm known \cite{clarkson2017low} for approximating all leverage scores has complexity $\widetilde{O}({\rm nnz}(A)+r^3)$, where ${\rm nnz}(A)$ is the number of nonzero entries of $A$ and $r$ is the rank of $A$. This is best possible for classical algorithms.

In this paper, we use the techniques of quantum linear algebra to accelerate the leverage score sampling. As an application, we propose a new quantum algorithm for rigid regressions with vector solution outputs.

\subsection{Main result}

Let $A$ be an $n\times d$ matrix of rank $r$ with singular value decomposition $A=UDV^T$, where $U$ is $n\times r$ consisting of the left singular vectors, $D$ is $r\times r$ consisting of the nonzero singular values, and $V$ is $d\times r$ consisting of the right singular vectors. 
The row leverage scores $\{\L_{R,1},\cdots,\L_{R,n}\}$ are defined as the squared norm of the rows of $U$. The distribution is defined by ${\rm Prob}(i) = \L_{R,i}/r$ for any $i\in\{1,\ldots,n\}$. We can also define column leverage scores according to $V$. 

To apply leverage score sampling on a quantum computer, one option is to prepare a quantum state corresponding to the leverage scores, which is defined as follows:
\be
\ket{\L_R} := \frac{1}{\sqrt{r}} \sum_{j=1}^n \sqrt{\L_{R,j}} \, \ket{j}.
\ee
Sampling according to leverage scores is equivalent to measuring this state in the computational basis. 
To prepare this state, we show that there is no need to compute any leverage score in advance, which can save a lot of computational time. Moreover, if we are interested in the value of some leverage scores, we can use the amplitude estimation technique \cite{brassard2002quantum}. 

In quantum computing, when dealing with matrix operations, we usually need an approach to encode the input matrix into a quantum circuit. A natural way is to construct a unitary with an efficient quantum circuit implementation such that the top-left corner is the input matrix. This is known as block-encoding \cite{chakraborty2019power, gilyen2019quantum}, see Definition \ref{defn:Block-encoding} for a rigorous statement.
Based on quantum singular value transformation \cite{gilyen2019quantum}, in this paper, we propose a quantum algorithm for preparing the state $\ket{\L_R}$.

\begin{thm}[Informal version of Theorem \ref{main theorem}]
\label{thm:intro 4}
Let $A$ be an $n\times d$ matrix of rank $r$. Assume that an $(\alpha,a,\epsilon)$ block-encoding of $A$ is constructed in cost $O(T)$, and the minimal nonzero singular value of $A$ is $\sigma_r$.
Then there is a quantum algorithm that prepares the state $\ket{\L_R}$ in cost $\widetilde{O}\left((T\alpha/\sigma_r) \sqrt{\min(n,d)/r} \right)$.
\end{thm}

In the above theorem, $\K:=T\alpha/\sigma_r$ can be viewed the cost of encoding a matrix into a quantum computer. When $\K$ is small, the quantum algorithm presented in Theorem \ref{thm:intro 4} achieves at least a quadratic speedup over classical algorithms \cite{clarkson2017low}. Furthermore, if the matrix has a full rank, the speedup can be exponential. In Proposition \ref{thm:lower bound} below, we will show that the dependence on $\sqrt{\min(n,d)/r}$ is tight. If we restrict ourselves in the framework of block-encoding, then the dependence on $T\alpha/\sigma_r$ is also tight. For convenience, we summarize the quantum/classical upper and lower bounds for leverage score sampling in Table \ref{table for LSS}.

\setlength{\arrayrulewidth}{0.3mm}
{\renewcommand
\arraystretch{1.5}
\begin{table}[h]
\centering
\begin{tabular}{|c|c|c|c|c|} 
 \hline
   & {\bf Upper bound} & {\bf Lower bound} \\ \hline
 {\bf Classical} & $O(nd+r^3)$ \cite{clarkson2017low} & $\Omega(n+d)$ (folklore) \\ \hline
 {\bf Quantum} & $\widetilde{O}(\K \sqrt{\min(n,d)/r})$ [Thm. \ref{main theorem}] & $\Omega(\K+ \sqrt{\min(n,d)/r})$ [Prop. \ref{thm:lower bound}] \\  \hline
\end{tabular}
\caption{A comparison of quantum/classical algorithms for leverage score sampling. The input matrix has size $n\times d$ and rank $r$.}
\label{table for LSS}
\end{table}
}

We want to emphasize that a big difference between  classical and quantum algorithms for leverage score sampling is that classically we first approximate all the leverage scores and then do the sampling, while quantumly we do the sampling first and then approximate the leverage scores of interest. This discrepancy can make a big difference in applications because we are usually more concerned about a small portion of the largest leverage scores.

As an application of Theorem \ref{thm:intro 4} and classical randomized algorithms, we propose a new quantum algorithm for solving the rigid (regularised) regression problems
\be
\label{intro:rigid-regression-new}
\argmin_{\x} \quad  Z(\x):=\|A \x - \b\|^2 + \lambda^2 \|\x\|^2
\ee
with the goal of outputting an approximate vector solution.
Here $\lambda> 0$ is the regularization parameter. 
Rigid regression is a useful method for a variety of problems in many different areas like machine learning, engineering, etc \cite{hoerl1970ridge,gruber2017improving,marquardt1975ridge}.
The regularization technique is often used to tackle ill-conditioned linear regressions and $\lambda$ is known as the regularization parameter \cite{tikhonov1963solution}.

The main problem we aim to solve is explicitly defined as follows. 

\begin{prob}
\label{intro:problem 1}
Let $A\in \mathbb{R}^{n\times d}, \b \in \mathbb{R}^{n}$, $\varepsilon\in[0,1]$, and let $\x_{\rm opt}$ be an optimal solution of (\ref{intro:rigid-regression-new}). The goal is to output a vector $\tilde{\x}_{\rm opt}$ such that
$
Z(\tilde{\x}_{\rm opt})
\leq (1+\varepsilon) 
Z(\x_{\rm opt}).
$
\end{prob}

\begin{thm}[Informal version of Theorem \ref{maim theorem 1}]
\label{intro-thm1}
Suppose $A$ has rank $r$, and an $(\alpha,a,\epsilon)$ block-encoding of $A$ is constructed in time $O(T)$. Then there is a quantum algorithm for Problem \ref{intro:problem 1}
in cost
\be
\widetilde{O}\left(
\frac{r}{\varepsilon} \left( \frac{T\alpha}{\lambda} \sqrt{(n+d)/\varepsilon} + d \right)
+ \frac{r^\omega}{\varepsilon^\omega} \frac{\|A\|^2}{\lambda^2} + r^3
\right),
\ee
where $\omega<2.373$ is the matrix multiplication exponent.
\end{thm}

In the low-rank case, the main cost of Theorem \ref{intro-thm1} comes from the first two terms, which is sublinear in $n$ and linear in $d$. 
For Problem \ref{intro:problem 1}, there is an obvious quantum algorithm. Namely, we first apply a quantum algorithm to obtain the quantum state of the solution and then use quantum tomography. The cost of this algorithm is $\widetilde{O}(\frac{T\alpha}{\lambda} \frac{d}{\varepsilon})$ under certain assumptions.\footnote{This algorithm is more suitable to return a solution $\tilde{\x}_{\rm opt}$ such that $\|\tilde{\x}_{\rm opt}-A^+\b\| \leq \varepsilon \|A^+\b\|$. Different from the algorithm in Theorem \ref{intro-thm1} which has a low-rank assumption, the complexity of this quantum algorithm is affected by the overlap of $\b$ in the column space of $A$. So these two algorithms are not directly comparable because of the different assumptions. For more, see the arXiv version of this paper \cite{shao2023improved}.} So Theorem \ref{intro-thm1} provides a different quantum algorithm for this problem which could be better when $n=\widetilde{O}(d^2)$ under the assumption that $A$ is low-rank.\footnote{In the low-rank case, it is possible to use the idea of quantum-inspired classical algorithms to propose a classical algorithm for Problem \ref{intro:problem 1}. Since the goal is to output a vector solution, the complexity is at least linear in the dimension. For quantum-inspired classical algorithms, it is possible to have a large dependence on $\|A\|_F/\sigma_r$, where $\|A\|_F$ is the Frobenius norm. Since this is not our main focus, we will not discuss this further.} The complexity of known classical algorithms for Problem \ref{intro:problem 1} is mainly dominated by ${\rm nnz}(A)=O(nd)$ in the low-rank case \cite{clarkson2017low,DBLP:conf/approx/AvronCW17}. By \cite[Lemma 48]{gilyen2019quantum}, in the worst case $T=O(\text{polylog}(n+d))$ and $\alpha=O(\sqrt{nd})$, so the quantum algorithm here is more efficient than these randomized classical algorithms.

\subsection{Related works in the quantum case}

In the quantum case, some progress has been made so far for the problem of performing leverage score sampling. In the thesis \cite{prakash2014quantum}, Prakash presented a quantum algorithm based on quantum singular value estimation for approximating leverage score distribution by assuming that the matrix is stored in an augmented QRAM data structure. Prakash's goal is similar to that of our paper, except that we used different techniques. Moreover, our result is indeed better. For example,  by only showing the dependence on $n,d,r$, Prakash's result is $O(\sqrt{n/r})$, while our result is $O(\sqrt{\min(n,d)/r})$. For linear regressions, usually $n\gg d$, this improvement can make a difference in certain applications. We also showed the tightness of the bound in this paper, which is missing in \cite{prakash2014quantum}.
In \cite{liu2017fast}, Liu and Zhang proposed a quantum algorithm for approximating all leverage scores by using quantum phase estimation for sparse matrices. They showed that approximating a single leverage score up to additive error $\varepsilon$ costs $\widetilde{O}(\kappa/\varepsilon)$, where $\kappa$ is the condition number of the input matrix. For leverage score sampling, the relative error is usually more desirable \cite{drineas2012fast}. It is possible that the smallest leverage score is close to $1/n$ (see Proposition \ref{prop:lower bound of leverage scores}), so to use their algorithm, we may need to choose $\varepsilon\approx 1/n$, which makes the overall complexity $\widetilde{O}(\kappa n)$. In the same scenario, our algorithm only costs $\widetilde{O}(\kappa \sqrt{n})$ by Proposition \ref{prop:lower bound of leverage scores} and Theorem \ref{main theorem}.

There is also some progress being made on the problem of computing vector solutions for linear regressions quantumly. A straightforward quantum algorithm is based on quantum tomography and quantum linear solvers, which is a subroutine of some quantum algorithms for optimization problems \cite{kerenidis2019q,kerenidis2020quantum}. A closely related one is Wang's algorithm. In \cite{wang2017quantum}, Wang studied the standard linear regression problem (i.e., $\lambda=0$) by outputting a vector solution $\tilde{\x}_{\rm opt}$ with the goal of $\|\tilde{\x}_{\rm opt} - \x_{\rm opt}\|_{\infty} \leq \varepsilon$. The complexity of the algorithm is $\widetilde{O}(d^{2.5}\kappa^3/\delta^2)$, where $\delta = \min(\varepsilon,1/d)$ and $\kappa$ is the condition number of $A$. 
In \cite{apers2020quantum}, Apers and de Wolf proposed a quantum algorithm for Laplacian and symmetric, weakly diagonally-dominant linear systems with certain polynomial speedups. 
In \cite{chen2021quantum}, Chen and de Wolf studied quantum algorithms for linear regression problems under $l_1$- or $l_2$-norm constraint. The $l_2$-norm constraint corresponds to rigid regression, which is the problem considered in our paper. For the $l_1$-norm constraint, they showed that quantum computers achieve quadratic speedup in terms of dimension, while for $l_2$-norm constraint there is no speedup in terms of dimension.

\subsection{Outline of this paper}

The paper is organised as follows.
In Section \ref{sec:Preliminaries}, we present some preliminary results on quantum computing and leverage scores that will be used in this paper.
In Section \ref{sec:Main results}, we apply quantum singular value transformation to speed up the leverage score sampling technique. 
In Section \ref{sec:application}, we present several quantum algorithms for rigid regression problems. 
In Section \ref{sec:Lower bound analysis}, we consider the quantum lower bounds of performing leverage score sampling and solving linear regressions.

{\bf Notation.}
We use $\{\e_1,\ldots,\e_n\}$ to denote the standard basis of $\mathbb{R}^n$, i.e., the $i$-th entry is 1 for $\e_i$. Using the Dirac notion, $\ket{i}=\e_i$. With $I_r$, we mean the $r$-dimensional identity matrix. For any matrix $A$, we use $A^+$ to denote its Moore-Penrose inverse, and $A^T$ to denote the transpose. The Frobenius norm $\|A\|_F$ is the square root of the sum of the absolute squares of the elements. The operator norm $\|A\|$ is the maximal singular value.
With notation $\widetilde{O}$, we ignore all polylog terms in the complexity. Given an integer $n$, we set $[n]=\{1,2,\ldots,n\}$.

\section{Preliminaries}
\label{sec:Preliminaries}

\subsection{Some necessary results on quantum computing}

The following is a useful quantum technique for approximating amplitudes in a given quantum state.

\begin{lem}[Amplitude estimation \cite{brassard2002quantum}]
\label{lem:Amplitude estimation}
There is a quantum algorithm which takes as input copies of a quantum state $\ket{\psi}$, a unitary transformation $U = 2\ket{\psi} \bra{\psi} - I$, a unitary transformation $V = I - 2P$ for some projector $P$, and an integer $M$. The algorithm outputs $\tilde{a}$, an estimate of $a = \bra{\psi}P\ket{\psi}$, such that
\bes
|a - \tilde{a}| \leq 2\pi \frac{\sqrt{a(1-a)}}{M} + \frac{\pi^2}{M^2}.
\ees
with probability at least $8/\pi^2$, and $O(M)$ uses of $U$ and $V$.
\end{lem}

\begin{defn}[Block-encoding \cite{chakraborty2019power}]
\label{defn:Block-encoding}
Suppose that $A$ is an $s$-qubit operator, $\alpha, \varepsilon \in \mathbb{R}^{>0}$ and $a\in \mathbb{N}$, then we say that the $(s+a)$-qubit unitary $U_A$ is an $(\alpha, a, \varepsilon)$ block-encoding of $A$, if
\be
\|A - \alpha (\bra{0}^{\otimes a}\otimes I) U_A (\ket{0}^{\otimes a}\otimes I) \| \leq \varepsilon,
\ee
where $\|\cdot\|$ is the operator norm. In matrix form, $U_A$ is a unitary such that its top-left corner is $A/\alpha$. This is also well-defined when $A$ is rectangular.
\end{defn}

The following result is very easy to prove. A proof is given in the arXiv version of this paper \cite{shao2023improved}.

\begin{lem}
\label{lem:Block-encoding of A tilde}
Given an $(\alpha, a, \varepsilon)$ block-encoding of $A$ that is constructed in cost $O(T)$, then  we can construct  an $(\alpha+\lambda, a+2, \varepsilon)$ block-encoding of $\widetilde{A}:=
\begin{pmatrix}
A \\
\lambda I \\
\end{pmatrix}$ in cost $O(T)$.
\end{lem}

The following result is a direct application of quantum singular value transformation. The proof is similar to that of Theorem 56 of \cite{gilyen2019quantum} for a specific function that approximates the sign function.

\begin{prop}
\label{lem:singular value threshold}
Assume that $U_A$ is an $(\alpha, a, \varepsilon)$ block-encoding of $A\in \mathbb{R}^{n\times d}$. Let $A=UDV^T$ be the singular value decomposition of $A$. The nonzero singular values of $A$ are $\sigma_1\geq \cdots \geq \sigma_r>0$.
Then there is an integer $m=O( (\alpha/\sigma_r) \log(1/\varepsilon))$ and a unitary $\widetilde{U}_A$, which is a $(1,a+1,4m\sqrt{\varepsilon/\alpha})$ block-encoding of $U V^T$.
Moreover, $\widetilde{U}_A$ can be implemented on a quantum circuit with $m$ uses of $U_A, U_A^T$, and $m$ uses of other one- and two-qubit gates.
\end{prop}

\begin{proof}
By Lemma 25 of \cite{gilyen2019quantum}, for any $\delta>0, \varepsilon\in(0,1/2)$, there is an odd polynomial $P(x)$ of degree $O(\delta^{-1} \log(1/\varepsilon))$ that approximates the sign function with the following property 
\[
\begin{cases}
|P(x)|\leq 1, & \forall x\in [-2,2], \\
|P(x)-{\rm sign}(x)| \leq \varepsilon, & \forall x\in [-2,2]\backslash (-\delta,\delta).
\end{cases}
\]
Now consider the following odd polynomial 
$Q(x) = (1-\varepsilon) \frac{P(x+2\delta)-P(-x+2\delta)}{2}.$
It satisfies
\[
\begin{cases}
|Q(x)| \leq 1 & \forall x\in[-1,1], \\
|Q(x)-{\rm sign}(x)|\leq 2\varepsilon, & \forall x\in[-1,-3\delta] \cup [3\delta,1], \\
|Q(x)|\leq 2\varepsilon, & \forall x\in[-\delta,\delta].
\end{cases}
\]
The construction here is similar to that of Lemma 29 of \cite{gilyen2019quantum}. 
Now we set $\delta = \sigma_r/3\alpha$.
By Theorem 17 of \cite{gilyen2019quantum} and note that $Q(x)$ is odd, we obtain a block-encoding of $Q^{{\rm (SV)}}(A/\alpha)\approx UV^T$.\footnote{For any matrix $A$ with SVD $A=UDV^T$, we define $Q^{{\rm (SV)}}(A)=UQ(D)V^T$.} The error $4m\sqrt{\varepsilon/\alpha}$ follows from the robustness analysis of QSVT, see Lemma 22 of \cite{gilyen2019quantum}.
\end{proof}

To ensure that the error term in the above result is small, we can choose $\varepsilon$ so that $4m\sqrt{\varepsilon/\alpha}\leq \tilde{\varepsilon}$ for some $\tilde{\varepsilon}$. For example, if $m/\sqrt{\alpha}$ is large, then we can choose
\footnote{We here used the fact that for any $\tau \geq 4$, if $x\geq 2\tau \log \tau$ then $x/\log x \geq \tau$. Regarding our problem,  we square both sides, then we obtain $ \varepsilon \log^2(1/\varepsilon) \leq \tilde{\varepsilon}^2\delta^2 / 16\alpha$, i.e., $64\alpha/\tilde{\varepsilon}^2\delta^2 \leq (1/\sqrt{\varepsilon})^{2} (\log (1/\sqrt{\varepsilon}))^{-2}$. So we can set $\tau = 8\sqrt{\alpha}/\tilde{\varepsilon}\delta$ and $x=1/\sqrt{\varepsilon}$, which leads to the claimed result.}
\be
\label{choice of epsilon 1}
\varepsilon \leq \frac{\tilde{\varepsilon}^2 \delta^2/\alpha}{256 } \frac{1}{\log^2(8\sqrt{\alpha}/\tilde{\varepsilon} \delta)}
=\widetilde{O}(\tilde{\varepsilon}^2 \delta^2/\alpha).
\ee
When constructing a block-encoding of $A$, the complexity is usually polylog in $1/\varepsilon$. So the overall complexity will not change too much even if we choose a much smaller $\varepsilon$ like (\ref{choice of epsilon 1}). 

In the end, we do some error analysis that will be used in the quantum speedup of leverage score sampling.  The matrix $U V^T$ defines a distribution state
$
\frac{1}{\sqrt{k}} \sum_{j=1}^d \ket{j} \otimes U V^T \ket{j}.
$
The probability of seeing $j$ is $P(j) = \|U V^T \ket{j}\|^2 / k$. Let $\widetilde{U}_A$ be the block-encoding constructed in Lemma \ref{lem:singular value threshold}, i.e.,
$
\widetilde{U}_A = \begin{pmatrix}
W & \cdot \\
\cdot & \cdot \\
\end{pmatrix}
$
and $\|W - U V^T \| \leq \tilde{\varepsilon}$. The matrix $W$ also defines a distribution state
$
\frac{1}{\|W\|_F} \sum_{j=1}^d \ket{j} \otimes W \ket{j}.
$
Denote the distribution as $\widetilde{P}$, i.e., $\widetilde{P}(j) = \|W \ket{j}\|^2 / \|W\|_F^2$.
We below consider the error of these two distributions. As we can imagine, we need to choose an appropriate $\tilde{\varepsilon}$ to ensure that the two distributions are close to each other. However, this will not cause too many problems if the dependence on $1/\tilde{\varepsilon}$ is polylog. Indeed, by the triangle inequality, we have
\[
\|P -\widetilde{P} \|_1 \leq \frac{2(2\sqrt{dk}+d\tilde{\varepsilon}) \tilde{\varepsilon}}{k - (2\sqrt{dk}+d\tilde{\varepsilon}) \tilde{\varepsilon}}.
\]
So to make sure the error of the two distributions is as small as possible, say $O(\hat{\varepsilon})$, we can choose $\tilde{\varepsilon}$ so that $(2\sqrt{dk}+d\tilde{\varepsilon}) \tilde{\varepsilon} = O(k\hat{\varepsilon})$, e.g.,
\be
\label{choice of epsilon 2}
\tilde{\varepsilon} = \hat{\varepsilon} \sqrt{k/d}.
\ee

\subsection{Leverage scores}

The following is the formal definition of (statistical) leverage scores.

\begin{defn}[Statistical leverage scores \cite{drineas2012fast}]
\label{defn:leverage scores}
Let $A$ be an $n\times d$ real-valued matrix of rank $r$. Let the singular value decomposition (SVD) of $A$ be $UDV^T$, where $U$ is $n\times r$ consisting of $r$ left singular vectors, $V$ is $d\times r$ consisting of $r$ right singular vectors, and $D$ is diagonal consisting of the nonzero singular values. The statistical leverage scores of the rows of $A$ are defined by
\be
\mathcal{L}_{R,j}(A) := \|\bra{j}U\|^2, \quad j \in [n], 
\ee
and of the columns of $A$ are defined by
\be
\mathcal{L}_{C,j}(A) := \|V^T \ket{j}\|^2, \quad j \in [d]. 
\ee
\end{defn}

When it makes no confusion, we sometimes just write $\mathcal{L}_{R,j}, \mathcal{L}_{C,j}$ for simplicity.
It is easy to see that the leverage scores are independent of scaling and $r = \sum_{j\in[n]} \L_{R,j} = \sum_{j\in[d]} \L_{C,j}$.
By assumption, we know that $U^T U = V^T V = I_r$. So the statistical leverage scores are also equal to
\bea
\mathcal{L}_{R,j} = \|\bra{j}UV^T\|^2, ~ j \in [n], \quad \text{ and } \quad
\mathcal{L}_{C,j} = \|UV^T \ket{j}\|^2, ~ j \in [d]. 
\label{equality for LSP}
\eea
The matrix $UV^T$ is known as the closest isometry to $A$ in the polar decomposition. 

The following lemma can be checked directly.

\begin{lem}
\label{lem:SVD of extended matrix}
Let $A$ be an $n\times d$ matrix with SVD $A=UDV^T$, where $U, V$ are square unitaries and $D$ is $n\times d$. Let $\widetilde{A}=
\begin{pmatrix}
A \\
\lambda I_d \\
\end{pmatrix}$, then the SVD of $\widetilde{A}$ is $\widetilde{A}=\widetilde{U} \Sigma^{-1} V^T$, where
\be
\widetilde{U} = \begin{pmatrix}
UD\Sigma \\
\lambda V\Sigma \\
\end{pmatrix} , \quad 
\Sigma = (D^TD + \lambda^2 I_d)^{-1/2}.
\ee
\end{lem}

We below consider a lower bound of leverage scores. The estimate will be used in the application below. Note that $r = \sum_{j\in [n]} \L_{R,j}$, so the average of row leverage scores is $r/n$. If we are only interested in the top $q=O(r/\varepsilon)$ nonzero row leverage scores (which is the case in many applications), then it is highly possible that they are all larger than $\Omega(r/nq) = \Omega(\varepsilon/n)$. We summarize the result in the following proposition rigorously.

\begin{prop}
\label{prop:lower bound of leverage scores}
Let $A$ be an $n\times d$ matrix with rank $r$.
Let $\{\L_{R,j}^{(S)}:j\in[n]\}$ be the sorted sequence of the row leverage scores in descending order.
Assume that $\varepsilon$ satisfies that 
$r/n-\varepsilon/2n \leq \varepsilon \leq 1- \varepsilon/2n$, 
then we have
$\mathcal{L}_{R,r/\varepsilon}^{(S)} \geq \varepsilon/2n$. 
\end{prop}

\begin{proof}
For simplicity, we denote $\L=\L_{R,1}^{(S)} = \max_j\L_{R,j}$. For any $\alpha\leq r/n$, we
denote $M_\alpha = \{j \in [n]: \mathcal{L}_{R,j} \geq \alpha\}$ and $m_\alpha = \#(M_\alpha)$.  Then
\[
r = \sum_{j\in [n]} \L_{R,j} 
\leq \sum_{j\in M_\alpha} \L + \sum_{j\not \in M_\alpha} \alpha
=m_\alpha \L +  (n-m_\alpha) \alpha.
\]
So $m_\alpha \geq (r-\alpha n)/(\L-\alpha)$. Based on this, for any small $\varepsilon$
\[
m_{\varepsilon/n} \geq \frac{r-\varepsilon}{\L-\varepsilon/n} \geq
\frac{r-\varepsilon}{1-\varepsilon/n} 
\approx r.
\]
We can further improve this bound as follows.
Set $  \varepsilon = \L-\alpha  $, then 
$m_{\alpha} \geq \frac{r-(\L-\varepsilon) n}{\varepsilon}.$
Set the lower bound as $\eta r/\varepsilon$ for some $\eta\in[0,1]$, then
$r-(\L-\varepsilon) n = \eta r$, i.e., $\alpha = \L-\varepsilon = (1-\eta)r/n$. Thus,
$
m_{(1-\eta)r/n} \geq \eta r/\varepsilon.
$
Since $r/n \leq \L\leq 1$ and $\varepsilon = \L - (1-\eta)r/n$, we then have 
$\eta r/n \leq \varepsilon \leq 1- (1-\eta)r/n$. 
Especially, if we choose $\eta = 1-\varepsilon/2r$, then $m_{\varepsilon/2n} \geq (1-\varepsilon/2r)r/\varepsilon = r/\varepsilon - 1/2$. Since $m_{\varepsilon/2n}$ is an integer, we have $m_{\varepsilon/2n} \geq r/\varepsilon$.
\end{proof}

The assumption in the above proposition is reasonable for us when $r\ll n$. In our analysis below, we are more interested in the first $O(r/\varepsilon)$ largest leverage scores. Also, if $\varepsilon < r/n$, then $r/\varepsilon > n$. This exceeds the number of leverage scores. Thus there is no need to consider this case. Moreover, suppose $k$ is the smallest index such $\mathcal{L}_{R,k}^{(S)} \leq \varepsilon/2n$. Then $\sum_{i=k}^n \mathcal{L}_{R,k}^{(S)}/r \leq \varepsilon(n-k)/2nr \leq \varepsilon/2r$. So if we generate $O(r/\varepsilon)$ samples according to the distribution $\{\L_{R,1}/r,\ldots,\L_{R,n}/r\}$, then with high probability, the corresponding leverage scores are all larger than $\varepsilon/2n$.


\section{Quantum speedups of leverage score sampling}
\label{sec:Main results}

Let $A$ be an $n\times d$ matrix with rank $r$ and SVD $A=UDV^T$, where $D$ is $r\times r$ consisting of all nonzero singular values. Our goal is to approximately prepare the following quantum states
\be
\label{leveral-score-sampling-states}
\ket{\mathcal{L}_R} := \frac{1}{\sqrt{r}} \sum_{j=1}^n \sqrt{\mathcal{L}_{R,j}}\, \ket{j}, \quad
\ket{\mathcal{L}_C} := \frac{1}{\sqrt{r}} \sum_{j=1}^d \sqrt{\mathcal{L}_{C,j}}\, \ket{j}
\ee
that contain the whole information of leverage scores of $A$ without computing them. 
In the proposed quantum algorithms below, we indeed obtain approximations of the following two states
\be 
\label{leveral-score-sampling-states:equivalent states}
\frac{1}{\sqrt{r}}\sum_{j=1}^n \ket{j} \otimes V U^T\ket{j}, \quad
\frac{1}{\sqrt{r}}\sum_{j=1}^d \ket{j} \otimes U V^T\ket{j}.
\ee
By (\ref{equality for LSP}), the above two states play exactly the same roles as $\ket{\mathcal{L}_R}, \ket{\mathcal{L}_C}$ in leverage score sampling. Hence, for convenience, we will not distinguish them in this paper. 
It is also easy to check that
\be
\label{LSP connection}
\frac{1}{\sqrt{r}}\sum_{j=1}^n \ket{j} \otimes V U^T\ket{j}
=
\frac{1}{\sqrt{r}}\sum_{j=1}^d U V^T\ket{j} \otimes \ket{j}.
\ee
By measuring the second register in the computational basis, we can perform column leverage score sampling. This means that it suffices to prepare one of the two states in (\ref{leveral-score-sampling-states:equivalent states}).

\begin{thm}
\label{main theorem}
Let $A$ be an $n\times d$ matrix of rank $r$. Assume that an $(\alpha,a,\epsilon)$ block-encoding of $A$ is constructed in cost $O(T)$. Let $\sigma_r$ be the minimal nonzero singular value of $A$.
\begin{enumerate}
\item There is a quantum algorithm that returns the states $\ket{\mathcal{L}_R}, \ket{\mathcal{L}_C}$ in cost 
\be
\label{complexity:new1}
\widetilde{O}\left( (T\alpha/\sigma_r) \sqrt{\min(n,d)/r} \right).
\ee

\item There is a quantum algorithm that for each $j$ returns $\widetilde{\mathcal{L}}_{R,j}, \widetilde{\mathcal{L}}_{C,j}$ in cost
$\widetilde{O}(T\alpha/\sigma_r\varepsilon)$ such that 
\be
\label{main thm:eq2}
\left|\widetilde{\mathcal{L}}_{R,j} - \mathcal{L}_{R,j}\right| \leq \varepsilon \sqrt{\mathcal{L}_{R,j} }+ \varepsilon^2,
\quad
\left|\widetilde{\mathcal{L}}_{C,j} - \mathcal{L}_{C,j}\right| \leq \varepsilon \sqrt{\mathcal{L}_{C,j} }+ \varepsilon^2.
\ee

\end{enumerate}
\end{thm}

\begin{proof}
(1). We assume that $n\geq d$ for simplicity. By (\ref{LSP connection}), it suffices to prepare $\ket{\L_C}$.
Let $U_A$ be the block-encoding of $A$, i.e., $U_A = \begin{pmatrix}
A/\alpha & \cdot \\
\cdot & \cdot \\
\end{pmatrix}.$
By Proposition \ref{lem:singular value threshold}, we can construct a block-encoding of $U V^T$ as follows $\widetilde{U}_A = \begin{pmatrix}
W & \cdot \\
\cdot & \cdot \\
\end{pmatrix}$
in cost $\widetilde{O}( T\alpha/\sigma_r)$. Here $\|U V^T - W\| \leq \varepsilon'$. Due to the error analysis in (\ref{choice of epsilon 1}), (\ref{choice of epsilon 2}), we choose
\[
\varepsilon' = 
\frac{{\varepsilon}^2 \sigma_r^2r}{256 \alpha d} \frac{1}{\log^2(8\sqrt{\alpha}d/{\varepsilon} \delta r)}.
\]
This ensures that the distribution defined by the square norm of columns of $W$ is $\varepsilon$ close to the distribution defined by $U V^T$ under the $l_1$-norm. Because of this, below for simplicity, we assume that $W = U V^T$.

To construct the distribution  state defined by the columns of $U V^T$, we consider the state $\frac{1}{\sqrt{d}} \sum_{j=1}^d \ket{j} \otimes \ket{0,j}.$
 Apply $\widetilde{U}_A$ to the second register $\ket{0,j}$, we then obtain
\beas
&& \frac{1}{\sqrt{d}} \sum_{j=1}^d \ket{j} \otimes \Big(\ket{0}\otimes UV^T \ket{j} + \ket{0}^\bot \Big) \\
&=& \frac{\sqrt{r}}{\sqrt{d}} \left( \frac{1}{\sqrt{r}}\sum_{j=1}^d \ket{j} \otimes  \ket{0}\otimes UV^T \ket{j}\right)
+{\rm orthogonal~terms}.
\eeas
By ignoring the ancilla qubit $\ket{0}$ in the middle, the state
\bes
\frac{1}{\sqrt{r}}\sum_{j=1}^d \ket{j} \otimes UV^T \ket{j}
\ees
corresponds to the probability distribution relating to column leverage scores. 
By amplitude amplification, we need to repeat $O( \sqrt{d/r})$ times.

(2). To estimate $\mathcal{L}_{C,j}$, we can just apply $\widetilde{U}_A$ to $\ket{0}\ket{j}$ and estimate the amplitude of $\ket{0}$ in the first qubit. To be more exact, $\widetilde{U}_A \ket{0}\ket{j} = \ket{0} \otimes W \ket{j} + \ket{1}\otimes \ket{G}$, where $\ket{G}$ is some garbage state. Then we can apply amplitude estimation to estimate the amplitude of $\ket{0}$, which is $\|W \ket{j}\|$ and it is $\varepsilon'$-close to $\mathcal{L}_{C,j}$. 
By Lemma \ref{lem:Amplitude estimation}, we can compute $\widetilde{\mathcal{L}}_{C,j}$ such that
$
|\widetilde{\mathcal{L}}_{C,j} - \|W \ket{j}\|| \leq \varepsilon \sqrt{\|W \ket{j}\|}+\varepsilon^2.
$
In this process, we need to use $\widetilde{O}(1/\varepsilon)$ many times of $\widetilde{U}_A$.
\end{proof}

In the classical setting, when using leverage score sampling, we usually need to approximate all the leverage scores first \cite{drineas2012fast}. Since there are $n$ (or $d$) leverage scores, the cost is at least linear in $n$ (or $d$). The best classical algorithm known so far costs $\widetilde{O}({\rm nnz}(A)+r^3)$ \cite{clarkson2017low}, where ${\rm nnz}(A)$ is the number of nonzero entries of $A$. In the quantum case, we can generate the distribution without computing the leverage scores. This can greatly reduce the complexity.
To be more exact, note that the dependence on the dimension in the result (\ref{complexity:new1}) is $\sqrt{\min(n,d)/r}$. An interesting case is when $n\gg d$ and $r=d$. This happens when solving linear regressions. In this case, the complexity is $\widetilde{O}(T\alpha/\sigma_r)$, which could be exponentially smaller than ${\rm nnz}(A)$.

In applications, we are usually more interested in the largest few leverage scores \cite{mahoney2011randomized}. Sampling according to the quantum states $\ket{\L_R}, \ket{\L_C}$ achieves this goal more naturally as sampling usually returns the indices corresponding to large leverage scores.
%
%
To obtain approximations of leverage scores up to relative error $\eta$ (which is usually a constant), we can set $\varepsilon$ as $\varepsilon=\eta 
\sqrt{\L_{R,j}}$ or $\varepsilon=\eta 
\sqrt{\L_{C,j}}$ for some small $\eta$. To make the complexity clear, we need a lower bound of the leverage scores, which is achieved by Proposition \ref{prop:lower bound of leverage scores} already.

In practice, we usually do not know the rank of the input matrix. So to apply Theorem \ref{main theorem}, we also need a rough estimate of the rank up to a small constant relative error.
We below consider this problem. Using diagonal matrix and quantum counting, it is not hard to see that the quantum lower bound of estimating rank up to relative error $\varepsilon$ is $\Omega(\varepsilon^{-1}\sqrt{\min(n,d)/r})$. In the proof of Theorem \ref{main theorem}, we can prepare 
\[
\frac{1}{\sqrt{d}} \sum_{j=1}^d \ket{j} \otimes \Big(\ket{0}\otimes UV^T \ket{j} + \ket{0}^\bot \Big)
\]
in cost $\widetilde{O}(\alpha T/\sigma_r)$. We can use amplitude estimation to approximate the amplitude of the first part, which is $r/d$, up to relative error $\varepsilon$. Thus, there is a quantum algorithm that computes the rank of a matrix in cost $O((\alpha T/\sigma_r\varepsilon) \sqrt{\min(n,d)/r})$.
In summary, we have the following result.

\begin{prop}[Approximating the rank of a matrix]
\label{prop:rank}
Assume $A\in \mathbb{R}^{n\times d}$ has rank $r$. Let $\sigma_r$ be the minimal nonzero singular value of $A$.
Given an $(\alpha,a,\epsilon)$ block-encoding of $A$ which is constructed in cost $O(T)$, then
there is a quantum algorithm that computes the rank up to relative error $\varepsilon$ in cost $O((\alpha T/\sigma_r \varepsilon)  \sqrt{\min(n,d)/r})$. Moreover, the dependence on $\varepsilon^{-1}\sqrt{\min(n,d)/r}$ is tight.
\end{prop}

For linear regression problems, the above complexity of estimating the rank is negligible in the quantum algorithms proposed below. 
Note that in \cite{belovs2011span}, Belovs applied the approach of span programs to determine if a matrix has a rank larger than a given integer $r$ on a quantum computer.
For this problem, Belovs proved a quantum lower bound of $\Omega(\sqrt{r(\min(n,d)-r+1)})$ based on the optimality of learning the Hamming-weight threshold function.
Moreover, there is also a quantum algorithm for this problem using $O(\sqrt{r(\min(n,d)-r+1)}LT)$ queries, where $L=\sqrt{\sum_{i=1}^r \sigma_i^{-2}/r} \leq 1/\sigma_r$ is a complexity measure, and $T$ is the cost of loading the matrix into a span program, which could be $\sqrt{nd}$ in the dense case. In Proposition \ref{prop:rank}, in the worst case when the matrix is dense, we have $\alpha=\sqrt{nd}$ and $T={\rm polylog}(n+d)$. So up to some polylog factors, these two results match in the worst case for the problem of computing the rank exactly.

\section{An application to rigid regressions}
\label{sec:application}

In this section, as an application of Theorem \ref{main theorem}, we consider the problem of solving linear regressions -- a problem of great concern to the classical and quantum communities. Regarding this problem, many quantum algorithms have been proposed so far \cite{harrow2009quantum,childs2017quantum,chakraborty2019power,shao2020row,chakraborty2022quantum}, to name a few here. These quantum algorithms, which only cost polylogarithmic in the dimension, usually return a quantum state of the solution. Recently, Apers and de Wolf in \cite{apers2020quantum} proposed an efficient quantum algorithm, with certain polynomial speedups, for Laplacian and symmetric, weakly diagonally-dominant linear systems. The output is a vector solution rather than a quantum state. 

We below consider the same problem but for more general linear systems on a quantum computer. More precisely, we focus on the rigid (or regularized) regression problem
\be
\label{rigid-regression-new}
\argmin_{\x\in \mathbb{R}^d} \quad  \|A \x - \b\|^2 + \lambda^2 \|\x\|^2,
\ee
where $A\in \mathbb{R}^{n\times d}, \b \in \mathbb{R}^{n\times 1}$. We also aim to obtain a vector solution. It is easy to see that the rigid regression problem (\ref{rigid-regression-new}) is equivalent to the following standard linear regression problem
\be
\label{intro:standard linear regression problem}
\argmin_{\x\in \mathbb{R}^d} \quad Z(\x)=\|\widetilde{A} \x - \tilde{\b}\|^2,
\ee
where
\be
\label{intro:new matrix}
\widetilde{A}:=
\begin{pmatrix}
A \\
\lambda I \\
\end{pmatrix},
\quad 
\tilde{\b} := \begin{pmatrix}
\b \\
0 \\
\end{pmatrix}.
\ee
Our basic idea for solving (\ref{rigid-regression-new}) is as follows, which combines ideas of randomized classical algorithms from \cite{drineas2011faster,chowdhury2018iterative,chen2015fast}.

\begin{itemize}
\item Construct a row sampling matrix $S$ of $A$ according to the row leverage scores of $\widetilde{A}$ given in (\ref{intro:new matrix}). Design a reduced rigid regression problem
\be
\label{intro:step1}
\argmin_{\x\in \mathbb{R}^d} \quad  \|SA \x - S\b\|^2 + \lambda^2 \|\x\|^2.
\ee
\item 
Construct a column sampling matrix $R$ of $A$ according to the column leverage scores of $SA$. Then approximate the solution with the estimator
\be
\label{intro:step2}
\tilde{\x}_{\rm opt} :=
(SA)^T \Big( (SAR) (SAR)^T + \lambda^2 I \Big)^{-1} S\b.
\ee
\end{itemize}
With Theorem \ref{main theorem}, we can implement the above two steps efficiently on a quantum computer. 
We below show more details.

\subsection{A preliminary result on rigid leverage score sampling}

Let $r$ be the rank of $A$.
Let the SVD of $A$ be $A=UDV^T$, where $D$ is $n \times d$. By Lemma \ref{lem:SVD of extended matrix} the SVD of $\widetilde{A}$ given in (\ref{intro:new matrix}) is $\widetilde{A}=\widetilde{U} \Sigma^{-1} V^T$, where
\be
\label{large svd}
\widetilde{U} = \begin{pmatrix}
UD\Sigma \\
\lambda V\Sigma \\
\end{pmatrix} , \quad 
\Sigma = (D^TD + \lambda^2 I_d)^{-1/2}.
\ee
We denote $\widehat{U} = UD\Sigma$ as the first $n$ rows of $\widetilde{U}$. Then it is easy to check that $\|\widehat{U}\|_F^2 = \sum_{i=1}^r \sigma_i^2/(\sigma_i^2+\lambda^2) \leq r$. The quantity $\|\widehat{U}\|_F^2$ is known as the {\em statistical dimension} of $A$, which will be denoted as ${\rm sd}_{\lambda}(A)$.

To solve rigid regression (\ref{rigid-regression-new}) by sampling and sketching, we should use the leverage scores associated with the rows of $\widehat{U}$. This is known as {\em rigid leverage score sampling}. In the quantum case, we need to construct the following quantum state (or an equivalent state) that contains the information of all rigid leverage scores
\[
\frac{1}{\|\widehat{U}\|_F} \sum_{i=1}^n \|\bra{i} \widehat{U}\| \, \ket{i}
=
\frac{1}{\|\widehat{U}\|_F} \sum_{i=1}^n \| V\widehat{U}^T \ket{i}\| \, \ket{i}.
\] 
This is a direct corollary of Theorem \ref{main theorem}.

\begin{prop}
\label{prop:leverage score sampling of U hat}
Let $A$ be an $n\times d$ matrix with rank $r$. Suppose that there is an $(\alpha,a,\epsilon)$ block-encoding of $A$ that is constructed in time $O(T)$. Let the SVD of $A=UDV^T$ with $D\in \mathbb{R}^{n\times d}$, and let $\widehat{U} = UD(D^TD+\lambda^2I_d)^{-1/2}$. Let ${\rm sd}_{\lambda}(A)=\|\widehat{U}\|_F^2$ be the statistical dimension of $A$. Assume that $\lambda=O(\|A\|)$ and $\lambda>0$.

\begin{enumerate}
\item There is a quantum algorithm that returns the state
\be \label{main theorem 2:eq1}
\frac{1}{\|\widehat{U}\|_F}\sum_{j=1}^n \ket{j} \otimes V \widehat{U}^T \ket{j}
\ee
in cost 
$
\widetilde{O}((T\alpha/\lambda) \sqrt{d/{\rm sd}_{\lambda}(A)}).
$
Moreover, an approximation of $\|\widehat{U}\|_F^2$ up to a constant relative error can be obtained within the same cost.

\item For any $j\in [n]$, there is a quantum algorithm that outputs $\widehat{\mathcal{L}}_{R,j}$ satisfying
\be
\Big| \widehat{\mathcal{L}}_{R,j} - \|\bra{j} \widehat{U}\|^2 \Big| \leq \varepsilon {\|\bra{j} \widehat{U}\|} + \varepsilon^2
\ee
in cost $\widetilde{O}(T\alpha/\lambda\varepsilon)$.
\end{enumerate}
\end{prop}

\begin{proof}
The second claim can be proved directly by applying Theorem \ref{main theorem} to $\widetilde{A}$. We below prove the first claim.
By Lemma \ref{lem:Block-encoding of A tilde}, we can construct an $(\alpha+\lambda,a+2,\epsilon)$ block-encoding of $\widetilde{A}$ in cots $O(T)$.
Note that $\widetilde{A}$ has rank $d$, so the quantum state corresponds to
the leverage scores of $\widetilde{U}$ is
\be
\label{leverage score sampling of U tilde}
\frac{\|\widehat{U}\|_F}{\sqrt{d}} \frac{1}{\|\widehat{U}\|_F} \sum_{i=1}^n \|\bra{i} \widehat{U}\| \, \ket{i}
+
\frac{\lambda \|V \Sigma\|_F}{\sqrt{d}} \frac{1}{\lambda \|V \Sigma\|_F} \sum_{i=n+1}^{n+d} \|\bra{i} \lambda V \Sigma\| \, \ket{i}.
\ee
By Theorem \ref{main theorem}, this state (or more precisely an equivalent state) is obtained in cost $\widetilde{O}(T(\alpha+\lambda)/\lambda) = \widetilde{O}(T\alpha/\lambda)$.
The first part of (\ref{leverage score sampling of U tilde}) corresponds to the leverage scores of $\widehat{U}$.
The probability of this part is
\[
\frac{1}{d} \sum_{i=1}^r \frac{\sigma_i^2}{\sigma_i^2+\lambda^2} = \frac{{\rm sd}_{\lambda}(A)}{d}.
\]
Thus, we can obtain this state in cost $ \widetilde{O}((T\alpha/\lambda)\sqrt{d/{\rm sd}_{\lambda}(A)})$ by amplitude amplification.
We can also use amplitude estimation to approximate the success probability ${\|\widehat{U}\|_F^2}/{d}$ up to a small constant relative error in cost
$\widetilde{O}( (T\alpha/\lambda) \sqrt{d/{\rm sd}_{\lambda}(A)} ).$
\end{proof}

\subsection{Solving standard linear regression problems}

Due to the connection between the rigid regression (\ref{rigid-regression-new}) and the standard linear regression (\ref{intro:standard linear regression problem}), in this subsection, we first focus on the standard linear regression problem.

Let $A\in \mathbb{R}^{n\times d}, \b \in \mathbb{R}^n$, we consider the following linear regression problem
\be \label{least-square problem}
\argmin_{\x \in \mathbb{R}^d} 
\quad \|A\x-\b\|.
\ee
Let $\x_{\rm opt}=A^+\b$ be an optimal solution, we aim to find an $\tilde{\x}_{\rm opt}$ so that
\[
\|A\tilde{\x}_{\rm opt}-\b\|
\leq (1+\varepsilon) \|A{\x}_{\rm opt}-\b\|.
\]
In \cite{drineas2011faster}, Drineas et al. proposed a randomized classical algorithm that finds an approximate solution of (\ref{least-square problem}) in cost $\widetilde{O}(nd+d^3/\varepsilon)$ when $e^d \geq n\geq d$. Since the computation or approximation of leverage scores is a bottleneck of classical algorithms, the algorithm of \cite{drineas2011faster} is not based on leverage score sampling. However, with slight modifications, the algorithm still works if we use leverage scores. We below state their algorithm in terms of the leverage scores as a comparison. 

\begin{breakablealgorithm}
\label{alg:linear system}
\caption{A randomized classical algorithm for solving linear regressions \cite{drineas2011faster}}
\begin{algorithmic}[1]
\REQUIRE An $n\times d$ real matrix $A$, an $n\times 1$ real vector $\b$ and an integer $q=O(r\log r+r/\varepsilon)$.
\ENSURE An approximate solution of $\arg\min_{\x} \|A\x-\b\|$.
\STATE Initialize $S \in \mathbb{R}^{q\times n}$ to be an all-zero matrix.
\STATE Approximate all leverage-scores $\{\L_{R,i}:i\in [n]\}$ of $A$ up to a small constant relative error. 
\STATE For $t\in[q]$, pick $i_t\in[n]$ with probability $p_{i_t}:=\L_{R,i_t}/r$, set the $i_t$-th row of $S$ as $\e_{i_t}^T/\sqrt{qp_{i_t}}$.
\STATE 
Solve the reduced linear regression problem $\arg\min_{\x} \|SA\x-S\b\|$ and return the solution.
\end{algorithmic}
\end{breakablealgorithm}

The main cost of Algorithm \ref{alg:linear system} comes from the second step of computing all leverage scores, which is $\widetilde{O}(nd)$ \cite{drineas2011faster} or  $\widetilde{O}({\rm nnz}(A)+{\rm rank}(A)^3)$ \cite{clarkson2017low}. In the last step, the induced problem is small-scale, which can be solved directly by Gaussian elimination or iteratively by the conjugate gradient descent method. With Theorem \ref{main theorem}, we can accelerate the above algorithm on a quantum computer.

\begin{breakablealgorithm}
\label{alg:linear system-quantum}
\caption{A quantum algorithm for solving linear regressions}
\begin{algorithmic}[1]
\REQUIRE An $n\times d$ real matrix $A$, an $n\times 1$ real vector $\b$ and an integer $q=O(r\log r+r/\varepsilon)$.
\ENSURE An approximate solution of $\arg\min_{\x} \|A\x-\b\|$.
\STATE Initialize $S \in \mathbb{R}^{q\times n}$ to be an all-zero matrix.
\STATE Apply Theorem \ref{main theorem} to generate the state $\ket{\mathcal{L}_R}$ of $A$.
\STATE Perform $q$ measurements to $\ket{\mathcal{L}_R}$, denote the results as $J\subseteq [n]$.
\STATE Apply Theorem \ref{main theorem} to approximate $p_j=\mathcal{L}_{R,j}/r$ for all $j\in J$ up to a small constant relative error. Denote the results as $\tilde{p}_j, j\in J$.
\STATE For $j\in J$, set the $j$-th row of $S$ as $\e_{j}^T/\sqrt{q\tilde{p}_j}$.
\STATE Solve the reduced linear regression problem $\arg\min_{\x} \|SA\x-S\b\|$ and return the solution.
\end{algorithmic}
\end{breakablealgorithm}

\begin{thm}
\label{thm:LS}
Let $A\in \mathbb{R}^{n\times d}, \b \in \mathbb{R}^n$. Suppose $A$ has rank $r$, and there is an $(\alpha,a,\epsilon)$ block-encoding of $A$ that can be constructed in time $O(T)$. Let $\sigma_r$ be the minimal nonzero singular value of $A$, and $q=O(r\log r+r/\varepsilon)$. Then with high probability, Algorithm \ref{alg:linear system-quantum} returns an approximate solution $\tilde{\x}_{\rm opt}$ in vector form such that 
\bea
\|A\tilde{\x}_{\rm opt}  - \b\| \leq (1+\varepsilon) \|A\x_{\rm opt} - \b\|
\eea
in cost
\be
\label{com: alg for linear system-quantum}
\widetilde{O}\left( \frac{\alpha T}{\sigma_r}
\frac{r \sqrt{n}}{\varepsilon^{1.5}} 
+ \min\left( \frac{rd\kappa}{\varepsilon}, \frac{rd^{2}}{\varepsilon}\right)
\right),
\ee
where 
$\kappa$ is the condition number of $A$. 
\end{thm}

\begin{proof}
The correctness follows from that of Algorithm \ref{alg:linear system}. We below estimate the complexity.
First, we can prepare $\ket{\mathcal{L}_R}$ by Theorem \ref{main theorem} in cost $\widetilde{O}( (\alpha T/\sigma_r) \sqrt{n/r})$. Then we perform $q=\widetilde{O}(r/\varepsilon)$ measurements and obtain $q$ samples. This means that we can perform the first three steps on a quantum computer in cost $\widetilde{O}( (\alpha T/\sigma_r)\sqrt{nr}/\varepsilon)$. 
In the fourth step, we approximate the leverage scores associated with the sampled $q$ rows up to a small constant relative error. This step costs $\widetilde{O}( (\alpha T/\sigma_r)r/(\varepsilon\sqrt{ \L_{R,q}^{(S)}}))$.  By Proposition \ref{prop:lower bound of leverage scores}, $\mathcal{L}_{R,q}^{(S)} = {\Omega}(\varepsilon/n)$, so the cost of this step is bounded by $\widetilde{O}(  (\alpha T/\sigma_r)\sqrt{n}r/\varepsilon^{1.5})$.
The last step is solving the reduced problem by conjugate gradient normal residual method, which costs $\widetilde{O}(qd\tilde{\kappa})=\widetilde{O}(rd \tilde{\kappa}/\varepsilon)$. Here $\tilde{\kappa}$ is the condition number of $SA$, which has order $O(\kappa)$, see \cite[Theorem 2.11]{woodruff2014sketching}.
%
%
In this last step, if we use a direct method (e.g., Gaussian elimination) to solve the reduced linear regression, then the complexity of the last step is $O(qd^2) = \widetilde{O}(rd^2/\varepsilon)$.
Putting it all together, we obtain the claimed result.
\end{proof}

Although the cost is quadratic in $d$, the structure of Algorithm \ref{alg:linear system-quantum} is simple. Denote $\mathcal{K} = \alpha T/\sigma_r$. In the low-rank case, the cost is $\widetilde{O}(\mathcal{K} \sqrt{n}/\varepsilon^{1.5} + d^2/\varepsilon)$, which could be better than the randomized classical algorithm of complexity $\widetilde{O}(nd)$ \cite{clarkson2017low}, especially when $n\gg d$. In the high-rank case, then the cost of Algorithm \ref{alg:linear system-quantum} is $\widetilde{O}(\mathcal{K} \sqrt{n}d/\varepsilon^{1.5}+d^3/\varepsilon)$. Especially, when $n\geq d^4$, the cost is dominated by the first term $\widetilde{O}(\mathcal{K} \sqrt{n}d/\varepsilon^{1.5})$. In comparison, the randomized classical algorithm costs $\widetilde{O}(nd)$ when $n\geq d^4$ \cite{clarkson2017low}. So when $\mathcal{K}$ is small, the quantum computer achieves a quadratic speedup with respect to $n$. In the main algorithm below, we will  reduce the dependence on $d$.

\subsection{Solving rigid regressions}

In this subsection, we come back to the rigid regression problem (\ref{rigid-regression-new}).
Note that if $\lambda$ is too large, say larger than $\|A\|$, then $\x=0$ could be a good approximation of the solution \cite[Lemma 13]{DBLP:conf/approx/AvronCW17}. So below, we assume that $\lambda = O(\|A\|)$. A natural randomized classical algorithm \cite{chowdhury2018iterative,chen2015fast} to solve rigid regression by the sampling and sketching method is choosing an appropriate column sampling matrix $R$ and computing
$\tilde{\x}_{\rm opt} = A^T (A R R^TA^T + \lambda^2 I_n)^{-1} \b.$

The following is a randomized classical algorithm for (\ref{rigid-regression-new}), which is effective when $n\ll d$.

\begin{breakablealgorithm}
\label{alg2:linear system-quantum}
\caption{A randomized classical algorithm for solving rigid regressions}
\begin{algorithmic}[1]
\REQUIRE An $n\times d$ real matrix $A$, an $n\times 1$ real vector $\b$ and an integer $c=O(r\log r+(r/\varepsilon)(\|A\|/\lambda)^2)$.
\ENSURE An approximate solution of $\arg\min_{\x} \|A \x - \b\|^2 + \lambda^2 \|\x\|^2$.
\STATE Initialize $R \in \mathbb{R}^{d\times c}$ to be an all-zero matrix.
\STATE Apply the randomized classical algorithm proposed in \cite{clarkson2017low} to approximate all the column leverage scores of $A$ up to a small constant relative error. Denote the results as $\widetilde{\L}_j, j\in [d]$.
\STATE Generate $c$ samples from the distribution $\{\tilde{p}_j:=\widetilde{\L}_j/r, j\in [d]\}$, set the results to be $J\subseteq [d]$.
\STATE For $j\in J$, set the $j$-th column of $R$ as $\e_{j}/\sqrt{q\tilde{p}_j}$.
\STATE Return $A^T (A R R^TA^T + \lambda^2 I_n)^{-1} \b.$
\end{algorithmic}
\end{breakablealgorithm}

\begin{prop}
\label{cor:fast classical alg}
Let $A\in \mathbb{R}^{n\times d}, \b \in \mathbb{R}^n$. Suppose $A$ has rank $r$. Assume that $\lambda = O(\|A\|)$ and $\lambda>0$. Then Algorithm \ref{alg2:linear system-quantum} returns an approximate solution $\tilde{\x}_{\rm opt}$ in vector form such that 
\bea
\|A \tilde{\x}_{\rm opt} - \b\|^2 + \lambda^2 \|\tilde{\x}_{\rm opt}\|^2 \leq (1+\varepsilon) (\|A {\x}_{\rm opt} - \b\|^2 + \lambda^2 \|{\x}_{\rm opt}\|^2)
\eea
in cost
\be
\widetilde{O}\left( 
{\rm nnz}(A) + r^3 + n^\omega + 
\frac{r}{\varepsilon} \frac{\|A\|^2}{\lambda^2} n^{\omega-1} \right),
\ee
where $\omega<2.373$ is the matrix multiplication exponent.
\end{prop}

\begin{proof}
The correctness is from \cite{chowdhury2018iterative,chen2015fast}. We next estimate the complexity.
The cost of approximating all leverage scores is $\widetilde{O}({\rm nnz}(A) + r^3)$.  Regarding the last step, we first solve $(ARR^TA^T+\lambda^2I_n)\x = \b$, then multiply $A^T$ and the solution $\x$. If $n$ is not large, we can solve the linear system directly, which costs $O(\min(n^{\omega-1} c, nc^{\omega-1}))$ for computing $ARR^TA^T+\lambda^2 I_n$ and costs $O(n^\omega)$ for solving the linear system $(ARR^TA^T+\lambda^2I_n)\x = 
\b$.\footnote{Here we used the result of multiplying rectangular matrices \cite{knight1995fast}.} The last step is multiplying $A^T$ and $\x$, which costs $O({\rm nnz}(A))$.
\end{proof}

Combining ideas from Algorithms \ref{alg:linear system-quantum} and \ref{alg2:linear system-quantum}, we are now ready to describe our main algorithm for solving (\ref{rigid-regression-new}).

\begin{breakablealgorithm}
\label{main alg}
\caption{A quantum algorithm for solving rigid regressions}
\begin{algorithmic}[1]
\REQUIRE An $n\times d$ real matrix $A$, an $n\times 1$ real vector $\b$, a regularization parameter $\lambda = O(\|A\|)$, and two integers $q,c$.
\ENSURE An approximate solution of $\arg\min_{\x} \|A\x-\b\|^2 + \lambda^2 \|\x\|^2$.
\STATE Initialize $S \in \mathbb{R}^{q\times n}$ to be all-zero matrices.
\STATE Apply Proposition \ref{prop:leverage score sampling of U hat} to generate the state (\ref{main theorem 2:eq1}).
\STATE Perform $q$ measurements to the state, denote the results as $J\subseteq [n]$.
\STATE Apply Proposition \ref{prop:leverage score sampling of U hat} to approximate $p_j=\|\bra{j} \widehat{U}\|^2/\|\widehat{U}\|_F^2$ for all $j\in J$ up to a small constant relative error. Denote the results as $\tilde{p}_j, j\in J$.
\STATE For $j\in J$, set the $j$-th row of $S$ as $\e_{j}^T/\sqrt{q\tilde{p}_j}$.
\STATE Apply Proposition \ref{cor:fast classical alg} to solve $\min_{\x} \|SA \x - S\b\|^2 + \lambda^2 \|\x\|^2$ with inputs $SA, S\b, c$.
\end{algorithmic}
\end{breakablealgorithm}

\begin{thm}
\label{maim theorem 1}
Let $A\in \mathbb{R}^{n\times d}, \b \in \mathbb{R}^n$. Suppose $A$ has rank $r$, and an $(\alpha, a, \epsilon)$ block-encoding of $A$ is constructed in cost $O(T)$. Let $q=\widetilde{O}(r/\varepsilon), c=\widetilde{O}(r\|A\|^2/\lambda^2\varepsilon)$, then with high probability Algorithm \ref{main alg} returns an approximate solution $\tilde{\x}_{\rm opt}$ in vector form such that
\bea
\label{maim theorem 1-promise}
\|A\tilde{\x}_{\rm opt}-\b\|^2 + \lambda^2 \|\tilde{\x}_{\rm opt}\|^2 \leq (1+\varepsilon) 
(\|A\x_{\rm opt}-\b\|^2 + \lambda^2 \|\x_{\rm opt}\|^2)
\eea
in cost
\be
\label{complexity of alg4}
\widetilde{O}\left(
\frac{r}{\varepsilon} \left( \frac{T\alpha}{\lambda} \sqrt{(n+d)/\varepsilon} + d \right)
+ \frac{r^\omega}{\varepsilon^\omega} \frac{\|A\|^2}{\lambda^2} + r^3
\right),
\ee
where $\omega<2.373$ is the matrix multiplication exponent.
\end{thm}


\begin{proof}
The correctness follows from Propositions \ref{prop:leverage score sampling of U hat}, \ref{cor:fast classical alg} and Theorem \ref{thm:LS}.
Regarding the complexity, from the complexity analysis in the proof of Theorem \ref{thm:LS}, the main cost comes from the approximation of the leverage scores. The overall cost of the first five steps is 
$\widetilde{O}((T\alpha/\lambda) r\sqrt{n+d}/\varepsilon^{1.5})$ based on Proposition \ref{prop:leverage score sampling of U hat}.\footnote{Although the complexities of preparing quantum states of the leverage scores are different in Theorem \ref{main theorem} and Proposition \ref{prop:leverage score sampling of U hat}, the difference is $\sqrt{r}$. This will not change the overall complexity since the main cost from the estimation of $c$ leverage scores by the complexity analysis in the proof of Theorem \ref{thm:LS}. } 
The cost of the last step is based on Proposition \ref{cor:fast classical alg}. Now, $A$ becomes $SA$ and $n$ becomes $q=\widetilde{O}(r/\varepsilon)$.
\end{proof}

The quantum part of Algorithm \ref{main alg} is the acceleration of leverage score sampling. If we replace it with the best classical algorithm known \cite{clarkson2017low}, then we obtain a randomized classical algorithm for rigid regression. To the best of the author's knowledge, this is not stated elsewhere, so we summarize this as a byproduct.

\begin{prop}
\label{prop:classical algorithm}
Making the same assumptions as Theorem \ref{maim theorem 1}. There is a randomized classical algorithm that achieves (\ref{maim theorem 1-promise})
in time
\be
\widetilde{O}\left(
{\rm nnz}(A) 
+\frac{r^\omega}{\varepsilon^\omega} \frac{\|A\|^2}{\lambda^2} + r^3
\right).
\ee
\end{prop}

\begin{proof}
The algorithm is similar to Algorithm \ref{main alg} except that the quantum algorithm for leverage score sampling is replaced with Clarkson-Woodruff's classical algorithm for approximating all leverage scores \cite{clarkson2017low}. In the complexity analysis, the cost $(T\alpha/\lambda)r\sqrt{(n+d)}/\varepsilon^{1.5}$ arises from the computation of leverage scores (which is now $O({\rm nnz}(A))$) and the term $rd/\varepsilon$ is absorbed in $O({\rm nnz}(A))$.
\end{proof}

Note that to achieve (\ref{maim theorem 1-promise}), previous algorithm costs $O({\rm nnz}(A) + rd^2/\varepsilon)$ when $n>d$ or costs $O({\rm nnz}(A) + (\|A\|/\lambda)^4 (n^3/\varepsilon^2))$ when $n\leq d$, see \cite{DBLP:conf/approx/AvronCW17}. In comparison, the algorithm stated in Proposition \ref{prop:classical algorithm} is better in the low-rank case.

\section{Lower bounds analysis}
\label{sec:Lower bound analysis}

In this section, we prove some lower bounds on performing leverage score sampling and solving linear regressions on a quantum computer. We show that there exist hard instances of solving these tasks even for quantum computers. In this section, when we say an oracle to query $A$ we mean the standard oracle of querying entries of $A=(A_{ij})$ in the way 
\be
\mathcal{O}: \ket{i,j} \ket{0} \mapsto \ket{i,j} \ket{A_{ij}}.
\label{oracle}
\ee
Let $U$ be a block-encoding of $A$, if $U$ is given to us, then we say that we can query $A$ through a block-encoding.

It is shown in Theorem \ref{main theorem} that the state $\ket{\L_R}$ can be prepared in cost $O( (\alpha T/\sigma_r)\sqrt{\min(n,d)/r})$. We below show a lower bound of $\Omega(\sqrt{\min(n,d)/r})$ in the query model and a lower bound of $\Omega(\alpha T/\sigma_r)$ in terms of the number of calls to the block-encoding of $A$.

\begin{prop}
\label{thm:lower bound}
Let $A$ be an $n\times d$ matrix of rank $r$. Assume that its minimal nonzero singular value is $\sigma_r$.

\begin{itemize}
\item 
If we query $A$ through the oracle (\ref{oracle}), then $\Omega(\sqrt{\min(n,d)/r})$ queries are required to make to prepare $\ket{\L_R}$.
\item  If we are given an $(\alpha,a,\epsilon)$ block-encoding of $A$, then $\Omega(\alpha/\sigma_r)$ many applications of this block-encoding are required to prepare $\ket{\L_R}$.
\end{itemize}
\end{prop}

\begin{proof} 
(1). To prove this claim,
we use the hardness of the unstructured search problem. Let $f:[n]\rightarrow \{0,1\}$ be a Boolean function with the promise that there is a subset $S\subseteq [n]$ of size $r$ such that
\[
f(x) = 
\begin{cases}
1 & x\in S, \\
0 & x\not\in S.
\end{cases}
\]
We are given an oracle to query $f$ and the goal is to find an element in $S$. It is known that the quantum query complexity of finding one $i\in S$ is $\Theta(\sqrt{n/r})$ by Grover's algorithm \cite{boyer1998tight}. We consider the matrix
\[
A = \sum_{k\in[n]} f(k) \ket{k} \bra{k} = \sum_{k\in S}  \ket{k} \bra{k},
\]
namely,
\[
A_{ij} = 
\begin{cases}
f(j) & i=j, \\
0    & \text{otherwise}.
\end{cases}
\]
Given the oracle to query $f$, we can also use it to query the entries of $A$.
It is easy to see that $A$ has rank $r$, and the quantum state of the row leverage scores of $A$ is
\[
\ket{\L_R} = \frac{1}{\sqrt{r}} \sum_{k\in S} \ket{k}.
\]
If we can prepare this state, we then can sample from it and find an $i\in S$.
Due to the optimaity of Grover's algorithm,  $\Omega(\sqrt{n/r})$ queries are required to prepare the state $\ket{\L_R}$.

(2). Regarding the dependence on $\alpha/\sigma_r$, we use the hardness of the problem of state discrimination. 
Given a pure state $\ket{\phi}$ known to be either $\ket{\psi_1}$ or $\ket{\psi_2}$, decide which is the case (here we ignore the global phases). It is known that if $\theta\in[0,\pi/2]$ is the angle between these two states, then $1/\theta$ copies of $\ket{\phi}$ are required \cite{Donnell}.
To use this result, we assume that there are two quantum states $\ket{\phi_1}, \ket{\phi_2}\in \mathbb{C}^n$ such that
$\ket{\psi_1} = (\sigma/\alpha) \ket{0} \ket{\phi_1} + \sqrt{1-(\sigma/\alpha)^2} \, \ket{1} \ket{\phi_2}$ and 
$\ket{\psi_2} = \ket{1} \ket{\phi_2}$. This is reasonable as we can always decompose $\ket{\psi_1}$ into a linear combination of $\ket{\psi_2}$ and its orthogonal part (up to a change of the orthogonal basis if necessary). If $\sigma/\alpha$ is small (which is the hard case), then the angle between these two states is close to $\sigma/\alpha$.

We consider the matrix
$\widetilde{A} = \begin{pmatrix}
    \sigma & 0 \\
    0 & A 
\end{pmatrix}$, where $A$ is promised to be either $\sigma\ket{\phi_1}$ or $0$. In the first case, the only nonzero singular value of $A$ is $\sigma$. 
Suppose we have an $(\alpha,a,\epsilon)$ block-encoding $U$ of $A$, then in case 1, the first column of $U$ must have the form of $\ket{\psi_1}$ for some states $\ket{\phi_1}, \ket{\phi_2}$. In case 2, the first column is $\ket{\psi_2}$. To distinguish these two cases, $\Omega(\alpha/\sigma)$ many applications of $U$ are required. 

It is easy to see that $\widetilde{U} = \begin{pmatrix}
    \sigma/\alpha & 0 & * \\
    0 & U & 0 \\
    * & 0 & *
\end{pmatrix}$ is an $(\alpha,a,\epsilon)$ block-encoding of $\widetilde{A}$ for some $a, \epsilon$.
Any block-encoding of $\widetilde{A}$ must have this form because of its block structure. 
Note that in case 1 (i.e., $A=\sigma\ket{\phi_1}$),
the quantum state of the leverage scores of $\widetilde{A}$ is 
$\frac{1}{\sqrt{2}}(\ket{0} + \sum_{j\in [n]} |\phi_{1j}| \ket{j})$, where $\ket{\phi_1} = \sum_{j\in [n]} \phi_{1j} \ket{j}$. In case 2, it is $\ket{0}$. If we measure the state, then we will see indices from $\{1,\ldots,n\}$ with probability  $1/2$ in case 1 and with probability 0 in case 2. From the results, we can distinguish which is the case, and so distinguish the first column of $U$.
Thus $\Omega(\alpha/\sigma)$ applications of $\widetilde{U}$ are required to prepare the quantum state of the leverage scores.
\end{proof}

In Algorithms \ref{alg:linear system-quantum}, \ref{alg2:linear system-quantum} and \ref{main alg}, the main cost indeed comes from the estimation of leverage scores. In Theorem \ref{main theorem}, the cost of approximating a leverage score up to additive error $\varepsilon$ is $\widetilde{O}(\alpha T/\sigma_r \varepsilon)$. Moreover, if we approximate it up to relative error $\varepsilon$, the complexity can be $\widetilde{O}((\alpha T/\sigma_r\varepsilon) \sqrt{n/r})$ by Proposition \ref{prop:lower bound of leverage scores}.
Below, we prove some lower bounds for estimating the largest leverage scores.

\begin{prop}
\label{prop:lower bound 2}
Let $A$ be an $n\times d$ matrix of rank $r$. Assume that its minimal nonzero singular value is $\sigma_r$. Suppose there is an oracle (\ref{oracle}) to query $A$.

\begin{itemize}
\item If we are also given an $(\alpha,a,\epsilon)$ block-encoding of $A$, then $\widetilde{\Omega}(\alpha /\sigma_r)$ many applications of the block-encoding are required to approximate the largest row leverage score up to a constant relative/additive error.

\item Suppose $r^2 2^{r^2}\geq n\geq r^5$ and we only query matrix $A$ via the oracle. Then any quantum algorithm that succeeds with probability at least $1-1/n$ must make $\widetilde{\Omega}(\sqrt{n/r})$ queries to approximate the largest row leverage score up to a constant relative error.
\end{itemize}

\end{prop}

\begin{proof}
(1).
We still use the hardness of the unstructured search problem. 
Given an oracle to query $\a=(a_1,\ldots,a_n) \in \{0,1\}^n$, the goal is to find an $i$ such that $a_i=1$ if exists. We assume that there is at most one such $i$ if exists. Any quantum algorithm requires making $\Omega(\sqrt{n})$ queries to $\a$.

We first consider the case that $\sigma_r=\Theta(1)$. Consider the matrix $A=(\varepsilon, a_1+\varepsilon,\ldots,a_n+\varepsilon)^T$, where $\varepsilon=1/\sqrt{n}$. 
Let $U$ be an $(\alpha,a,\epsilon)$ block-encoding of $A$ that is constructed in time $O(T)$, then $\alpha T = \Omega(\sqrt{n})$.
Indeed, the first column of $U$ is
\[
\frac{1}{\alpha} \ket{0} \left( \varepsilon \ket{0} +
\sum_{i=1}^n (a_i+\varepsilon) \ket{i}
\right) + \ket{1} \ket{G},
\]
where $\ket{G}$ is a garbage state.
It is easy to see that $\|A\|_F = \Theta(1)$. So by measuring the above state in the computational basis, we will see $i$ such that $a_i=1$ (if exists) with probability $\Theta(1/\alpha^2)$. The cost of finding such an $i$ is $O(\alpha T)$ by amplitude amplification.
On the other hand, when we have an oracle to query $A$, we can construct a $(\sqrt{n},3+\log n,\epsilon)$ block-encoding of $A$ in cost $T=O({\rm poly} \log n)$ due to its density, e.g., see \cite[Lemma 48]{gilyen2019quantum}. So $\alpha T= \widetilde{\Theta}(\sqrt{n})$. 

Suppose $i$ exists,
then the quantum state of the leverage scores for this matrix is 
\[
\frac{1}{\sqrt{(1+\varepsilon)^2+1}} \left( \varepsilon \ket{0} +
(1+\varepsilon)  \ket{i} + \varepsilon \sum_{j\in [n]\backslash \{i\}} \ket{j}
\right).
\]
The leverage scores are $(1+\varepsilon)/\sqrt{(1+\varepsilon)^2+1}\approx 1$ and $\varepsilon/\sqrt{(1+\varepsilon)^2+1} \approx \varepsilon=1/\sqrt{n}$.
If no such $i$, then the quantum state of the leverage scores is 
\[
\frac{1}{\sqrt{n+1}} \left(
\ket{0} + \sum_{j=1}^n \ket{j}
\right).
\]
The nonzero leverage scores are all equal to ${1}/{\sqrt{n+1}} \approx \varepsilon$. 
If we can estimate the largest leverage score
up to a small constant additive/relative error, then we can determine if there is an $i$ such that $a_i=1$. From this, we can solve the unstructured search problem. Thus, estimating the largest leverage score costs at least $\widetilde{\Omega}(\alpha T)$.

Generally, we can prove a lower bound of $\Omega(\alpha T/\sigma_r)$ by expanding $A$ into $r$ orthogonal columns. More precisely, assume that $rs=n$ for convenience. We decompose $\a$ into $s$ equal parts $\a=(\a_1,\ldots,\a_r)$, each part contains $s$ elements. We view each $\a_i$ as a row vector and define an $n\times r$ matrix by setting the first column as $(\varepsilon, \a_1+\varepsilon, \0_s,\0_s,\ldots,\0_s)^T$, the second column as $(0,\0_s,\a_2+\varepsilon,\0_s,\ldots,\0_s)^T$, and so on. The last column is $(0,\0_s,\0_s,\ldots,\0_s,\a_r+\varepsilon)^T$. Here $\0_s$ is a zero vector of dimension $s$ and $\a_i+\varepsilon$ means that we add each entry of $\a_i$ by $\varepsilon$. So $A$ has orthogonal columns. The minimal nonzero singular value is $\sigma_r=\|\a_i+\varepsilon\|$ for some $i$, i.e., $\sigma_r=\sqrt{s}\varepsilon=\sqrt{s/n}$. 

Using a similar argument to the above, we can show that $\alpha T = \widetilde{\Theta}(\sqrt{n/r})=\widetilde{\Theta}(\sqrt{s})$. Indeed, let $U$ be an $(\alpha,a,\epsilon)$ block-encoding of $A$ that is constructed in time $O(T)$. We consider the sum of the first $r$ columns of $U$, i.e., we apply $U$ to $\ket{0} \otimes \frac{1}{\sqrt{r}}\sum_{j\in [r]} \ket{j}$ to obtain
\[
\frac{1}{\alpha \sqrt{r}} \ket{0} \left( \varepsilon \ket{0} +
\sum_{i=1}^n (a_i+\varepsilon) \ket{i}
\right) + \ket{1} \ket{G}.
\]
The probability of seeing $i$ such that $a_i=1$ (if exists) is $\Theta(1/\alpha^2 r)$. So $\sqrt{r}\alpha T =\Omega(\sqrt{n})$, i.e., $\alpha T = \Omega(\sqrt{n/r}) = \Omega(\sqrt{s})$. Now $A$ is $\Theta(s)$ sparse, so by
\cite[Lemma 48]{gilyen2019quantum}, we can construct a $(\sqrt{s},3+\log n, \epsilon)$ block-encoding of $A$ in time $T=O({\rm poly} \log n)$. Thus, $\alpha T = \widetilde{\Theta}(\sqrt{s})$.

Suppose $i$ exists (for notational convenience, assume $a_1=1$), then the leverage scores are
\[
\frac{1+\varepsilon}{\sqrt{s\varepsilon^2 + (1+\varepsilon)^2}} \approx 1, 
\,\,
\frac{\varepsilon}{\sqrt{s\varepsilon^2 + (1+\varepsilon)^2}} \approx \varepsilon = \frac{1}{\sqrt{n}},
\,\,
\frac{1}{\sqrt{s}}.
\]
If no such $i$, then the leverage scores are
\[
\frac{1}{\sqrt{s+1}},
\,\,
\frac{1}{\sqrt{s}}.
\]
Hence, if we can estimate the largest leverage score
up to a small constant additive/relative error, then we can determine if there is an $i$ such that $a_i=1$. From this, we can solve the unstructured search problem. Thus, estimating the largest leverage score costs at least $\widetilde{\Omega}(\alpha T/\sigma_r)$.
In this construction, we have $\alpha T/\sigma_r = \widetilde{\Theta}(\sqrt{n})$.

(2). We will use the second construction given above to prove this claim. In that construction, the matrix has rank $r$. Now we assume that $\a$ has Hamming weight $r$, i.e., we assume that there are $r$ marked items. By Grover's algorithm, to find one marked item, $\Omega(\sqrt{n/r})$ queries are required. Without loss of generality, we assume that $1\leq r\ll n$. We below show that we can use the binary search method to find one marked item if we can approximate the largest leverage score. We decompose $\a$ into two equal but disjoint parts $\a', \a''$. For $\a'=(\a'_1,\ldots,\a'_r)$, we construct a matrix in the same as that of $\a$. We denote it as $A'$. If $\a'$ contains marked items, then the largest leverage score of $A'$ is
\[
\frac{1+\varepsilon}{\sqrt{s'\varepsilon^2 + r'(1+\varepsilon)^2}} \approx \frac{1}{\sqrt{r'}},
\]
where $s'=n/2r$ and $1\leq r'\leq r$. Here $r'$ stands for the number of marked items in $\a_1'$. If there are no marked items in $\a'$, then the largest leverage score is $1/\sqrt{s'} = \sqrt{2r/n}$. If we can approximate the largest leverage score up to a small constant relative error, then we can distinguish which is the case.
If $\a'$ has marked items, we continue the bipartition; otherwise, we do the bipartition for $\a''$.   Note that if we repeat $l$ times, the largest leverage score is either $1/\sqrt{r'}$ or $\sqrt{2^l r/n}$. So after $l=\log (n/r^2)$ repetitions, these two leverage scores are hard to separate. However, in this case, there are $n/2^l = r^2$ items left. Now we can just query all the items because $r$ is small. By assumption, the success probability of $i$-th step is $1-2^i/n$. So the success probability of the above procedure is $\prod_{i=1}^l (1-2^i/n) \geq (1-1/r^2)^l = \Theta(1)$ when $r^2 \geq l = \log (n/r^2)$. Hence, the lower bound is $\Omega(\sqrt{n/r}-r^2) = \Omega(\sqrt{n/r}) $ when $n\geq r^5$.
\end{proof}

Finally, we show the hardness of solving linear regressions on a quantum computer.

\begin{prop}
\label{prop3:lower bounds}
For the linear regression problem $\arg\min\|A\x-\b\|$ with $A\in \mathbb{R}^{n\times d}, \b \in \mathbb{R}^n$,

\begin{itemize}
\item any quantum algorithm that returns $\tilde{\x}_{\rm opt}$ such that $\|\tilde{\x}_{\rm opt} - A^+\b\| \leq \varepsilon \|A^+\b\|$
needs to at least make 
$\Omega(\sqrt{n}+d)$ queries to $A$ or
$\Omega(\sqrt{n}+d)$ queries to $\b$. 
\item any quantum algorithm that returns $\tilde{\x}_{\rm opt}$ and $\|A\tilde{\x}_{\rm opt}-\b\|$ such that $\|A\tilde{\x}_{\rm opt}-\b\| \leq (1+\varepsilon) \min_{\x} \|A\x-\b\|$ needs to at least make $\Omega(\sqrt{n}+d)$ queries to $\b$. 
\end{itemize}

\end{prop}

\begin{proof}
The dependence on $d$ is obvious since $\tilde{\x}_{\rm opt}$ has $d$ entries. Below we focus on the dependence on $n$, which is usually much larger than $d$ for linear regressions.
We consider the existence problem. Given an oracle to query $\a=(a_1,\ldots,a_n) \in \{0,1\}^n$, determine if there is an $i$ such that $a_i=1$. This is equivalent to computing the OR function \cite{beals2001quantum}, so any quantum algorithm requires making $\Omega(\sqrt{n})$ queries to $\a$. Now let $\b\in \mathbb{R}^{n+1}$ be the all one vector and 
\[
A = \begin{pmatrix}
1 & 0 \\
0 & \a^T
\end{pmatrix}_{(n+1)\times 2}.
\]
If no such $i$, then $A^+ \b = \begin{pmatrix}
1 \\
0 
\end{pmatrix}$.
Otherwise, $A^+ \b = \begin{pmatrix}
1 \\
1 
\end{pmatrix}.$
These two states are clearly separated. So from the output, we can determine if there is an $i$ such that $a_i=1$. 
Thus outputting $\tilde{\x}_{\rm opt}$ with $\|\tilde{\x}_{\rm opt} - A^+\b\| \leq \varepsilon \|A^+\b\|$
requires at least making $\Omega(\sqrt{n})$ queries to $A$. 

We next consider the following construction. Let
\beas
A = \frac{1}{\sqrt{n}} \sum_{i\in [n]} \ket{i},
\quad 
\b = \frac{1}{\sqrt{n}} \sum_{i\in [n]} \ket{i} + \sqrt{n} \sum_{i\in[n]} a_i \ket{i}.
\eeas
If there is no $i$ such that $a_i=1$, then the optimal solution is $x_{\rm opt} = 1$ and the optimal value is $Z = 0$. If there is one $i$ such that $a_i=1$, then $x_{\rm opt} = 2$ and the optimal value is $Z = \sqrt{n-1}$. Thus, if we can compute $\tilde{x}_{\rm opt}$ such that $|\tilde{x}_{\rm opt} - x_{\rm opt}| \leq \varepsilon |x_{\rm opt}|$, then we can also solve the unstructured search problem. Moreover, if the quantum algorithm is required to output an approximation of the optimal value, then this algorithm must query $\Omega(\sqrt{n})$ entries of $\b$.
\end{proof}

Combining Theorem \ref{maim theorem 1} and Proposition \ref{prop3:lower bounds}, for solving linear regressions, it is possible that quantum computers have credible quadratic speedups in terms of $n$ (i.e., the number of constraints), which is usually much larger than $d$ (i.e., the dimension) for linear regressions. 

Recall that if we query $A$ through an $(\alpha,a,\epsilon)$ block-encoding, then the cost of solving linear regressions with quantum state output is $\widetilde{O}(\alpha \kappa T)$ \cite{chakraborty2019power}. It is known that the dependence on $\kappa$ is believed to be optimal \cite{harrow2009quantum}. In the first construction of the above proof, we indeed showed that the dependence on $\alpha$ is also optimal. To see this, in the first construction, it is easy to see that the condition number of the matrix is 1 and the two solution states can be used to determine if there is a marked item. In addition, we can  construct a $(\sqrt{n},3+\log n, \epsilon)$ block-encoding in time polylog in $n$ by \cite[Lemma 48]{gilyen2019quantum}. Although the matrix is sparse now, we do not know the nonzero positions. So the result \cite[Lemma 48]{gilyen2019quantum} should be used by viewing $A$ as a dense matrix. Thus $\alpha=\sqrt{n}$. This means obtaining the quantum state of the solution requires making $\Omega(\alpha)$ queries to $A$. In the construction, the angle of the two solution states is $\pi/4$, so the allowable error to separate them is the distance between $\ket{0}$ and $\cos(\pi/8) \ket{0} + \sin(\pi/8) \ket{1}$, which is about 0.39. In summary, we conclude above as the following result.

\begin{prop}
\label{prop4:lower bounds}
For the linear regression problem $\arg\min\|A\x-\b\|$ with $A\in \mathbb{R}^{n\times d}, \b \in \mathbb{R}^n$, if we query $A$ through an $(\alpha,a,\epsilon)$ block-encoding, then 
$\Omega(\alpha)$ applications of this block-encoding are required to prepare $\ket{A^+\b}$ up to any error $\delta \le 0.39$.

\end{prop}

Note that in the second construction of the proof of Proposition \ref{prop3:lower bounds}, we have $\alpha=1$ while the lower bound is still $\Omega(\sqrt{n})$. This means that to solve linear regressions efficiently on a quantum computer, we need to take both $A$ and $\b$ into account. The results of \cite{boutsidis2013near} might be helpful in improving quantum algorithms.

\section{Conclusions}

In this work, we proposed a quantum algorithm for accelerating leverage score sampling. Due to the wide applications of leverage score sampling, it is possible that Theorem \ref{thm:intro 4} can be used to speed up other applications in randomized numerical linear algebra. 
It is also interesting to know if we can propose other quantum algorithms not using QSVT for leverage score sampling, hoping to lead to more convincing quantum speedups, e.g., speed up the algorithm \cite{drineas2006fast3}.
For rigid regressions, the main quantum techniques (i.e., block-encoding and QSVT) used in this paper are very general, so we may not expect clean results like \cite{apers2020quantum,chen2021quantum}.
To explore more credible quantum speedups, one question is can we remove $T\alpha/\lambda$ in Theorem \ref{intro-thm1}? From the lower bounds analysis in Section \ref{sec:Lower bound analysis}, some new ideas are required. Also based on the lower bound analysis in Proposition \ref{prop3:lower bounds}, it is interesting to know if we can propose a quantum algorithm that can achieve credible quadratic speedups with respect to $n$.

\subsection*{Acknowledgements}

I would like to thank Ashley Montanaro, Yuji Nakatsukasa for their valuable discussions. I also would like to thank the anonymous referees of TQC2023 and QCTIP 2023 for the useful comments.
I acknowledge support from EPSRC grant EP/T001062/1. This project has received funding from the European Research Council (ERC) under the European Union's Horizon 2020 research and innovation programme (grant agreement No.\ 817581). No new data were created during this study.

\bibliographystyle{plain}
\bibliography{main}

\end{document}